\documentclass[10pt, conference, letterpaper]{IEEEtran}
\usepackage{amsmath,amssymb,amsfonts,amsthm}
\usepackage{algorithm,algorithmic}
\usepackage{graphicx}
\usepackage{subcaption}
\usepackage{textcomp}
\usepackage{xcolor}
\usepackage{adjustbox}
\usepackage{multirow}
\usepackage{balance}
\usepackage{pbox}
\def\BibTeX{{\rm B\kern-.05em{\sc i\kern-.025em b}\kern-.08em
    T\kern-.1667em\lower.7ex\hbox{E}\kern-.125emX}}
\begin{document}

\title{The Effect of Ground Truth Accuracy on the Evaluation of Localization Systems}

\author{
\IEEEauthorblockN{Chen Gu}
\IEEEauthorblockA{\textit{Google, USA} \\
guc@google.com}
\and
\IEEEauthorblockN{Ahmed Shokry}
\IEEEauthorblockA{\textit{Alexandria University, Egypt} \\
ahmed.shokry@alexu.edu.eg}
\and
\IEEEauthorblockN{Moustafa Youssef}
\IEEEauthorblockA{\textit{AUC and Alexandria University, Egypt} \\
moustafa@alexu.edu.eg}
}

\maketitle

\begin{abstract}
The ability to accurately evaluate the performance of location determination systems is crucial for many applications. Typically, the performance of such systems is obtained by comparing ground truth locations with estimated locations. However, these ground truth locations are usually obtained by clicking on a map or using other worldwide available technologies like GPS. This introduces ground truth errors that are due to the marking process, map distortions, or inherent GPS inaccuracy.

In this paper, we present a theoretical framework for analyzing the effect of ground truth errors on the evaluation of localization systems. Based on that, we design two algorithms for computing the real algorithmic error from the validation error and marking/map ground truth errors, respectively. We further establish bounds on different performance metrics.

Validation of our theoretical assumptions and analysis using real data collected in a typical environment shows the ability of our theoretical framework to correct the estimated error of a localization algorithm in the presence of ground truth errors. Specifically, our marking error algorithm matches the real error CDF within 4\%, and our map error algorithm provides a more accurate estimate of the median/tail error by 150\%/72\% when the map is shifted by 6m.
\end{abstract}


\section{Introduction}
Location determination technologies have gained momentum recently with a number of outdoor and indoor applications such as directions finding, directed ads, E911, driverless cars, among others. Typically, the performance of location determination systems is quantified through some measures, e.g. the full error CDF, or a single-value metric such as median, mean, or different percentiles. To do that, a ground truth data set is usually collected tagged with the actual user location to compare the estimated location against. This ground truth location can be entered manually, by a user clicking on the map where he/she is standing, or automatically; by using a higher accuracy tracking technology; e.g. using GPS or GNSS systems as ground truth for cellular-based localization technologies
\cite{shokry2018deeploc, rizk2019effectiveness, elbakly2019crescendo, ibrahim2011cellsense, aly2017accurate,paek2011energy}.

These methods for ground truth collection are inherently noisy and may lead to inaccuracy in the reported evaluation results of a particular localization system
\cite{shokry2020dynamicslam, abbas2019wideep, kotaru2015spotfi, shokry2017tale, aly2017accurate, aly2013dejavu, ibrahim2013enabling, gjengset2014phaser, vasisht2016decimeter}. 
In particular, a human ground truth collector may click on the map in the wrong location when marking his/her position. This is especially true in open areas, when there are no landmarks to help determine where the user is standing, or when the user marks the location on the limited mobile device screen. Moreover, in many cases, ground truth locations may be marked at a lower granularity to reduce the collection overhead, e.g. by marking the start and end location of a trace and interpolating the ground truth points between them \cite{paek2011energy, cherntanomwong2009signal}. Furthermore, the map used for marking the ground truth locations may contain errors itself, e.g. have an offset or scale error (Figure~\ref{fig:japan}), an error whose possibility increases with a worldwide deployment. All of these factors could affect the reported accuracy of a localization system, which can be significant when the current state-of-the-art is trying to squeeze the errors in their systems. %

In this paper, we present a theoretical framework for analyzing the effect of ground truth errors on the evaluation of localization systems. Using this framework, we design two algorithms for computing the real algorithmic error from the validation error and marking/map ground truth errors, respectively. %
These algorithms can be used to adjust the reported accuracy on the different metrics of a given localization system to account for ground truth collection issues, as well as to prioritize which factors should be handled more carefully. Specifically, our analysis shows interesting findings: (a) The impact of marking error on evaluation accuracy is quadratic; (b) As the accuracy of the localization system gets better, the impact of error in ground truth grows; (c) The 95\%-tail error is at least twice of the median/mean errors; (d) Marking error has more impact than map error on the mean and median, but less impact on the tail.

We validate our assumptions and analysis on a real indoor WiFi dataset. The results show the ability of our theoretical framework to correct the estimated error of a localization algorithm to obtain a more realistic error in the presence of ground truth errors. In particular, our marking error algorithm matches the real error within 4\% in all percentiles. Moreover, our map error algorithm provides a more accurate estimate of the median/tail error by more than 20\%/5\% when the map is shifted by 1m in X and Y directions. This enhancement increases to 150\%/72\% when the map is shifted by 6m.

To the best of our knowledge, this is the first work that quantifies the effect of ground truth accuracy on the performance of location determination systems.

The rest of the paper is organized as follows: Section \ref{sec:model} presents the mathematical model formulation and notation used in the paper. Section \ref{sec:theory} gives the details of our theoretical analysis followed by the experimental validation of the theoretical results in Section \ref{sec:evaluation}. Section \ref{related_work} discusses related work. Finally, we conclude the paper in section \ref{conc}.

\begin{figure}[t]
\centerline
{\includegraphics[width=0.25\textwidth]{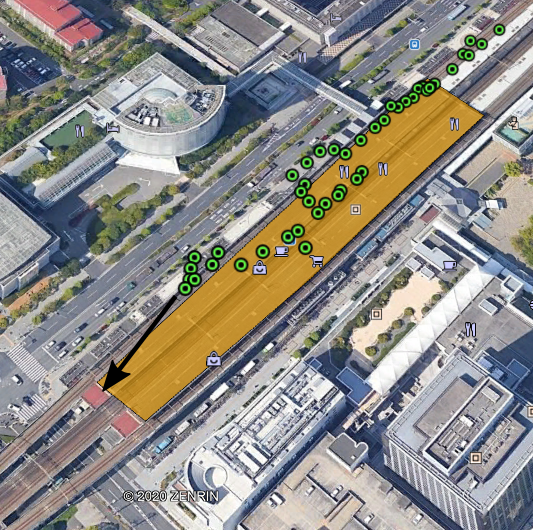}}
\caption{Map error example. %
The human-marked ground truth locations (green circles) inside a railway station differ by a consistent offset of about 30 meters in the northeast direction. This is because %
the used map system had an incorrect alignment of the building.
}
\label{fig:japan}
\end{figure}

\section{Fundamental Mathematical Model}
\label{sec:model}
\newtheorem{example}{Example}

In this section, we introduce the mathematical model used for evaluating the performance of a general location determination system and state the problems we study in this paper.

\subsection{Types of Ground Truth Errors}
First, we define two types of ground truth error: marking error and map error. In Figure \ref{fig:errors}(a), the user is at the green circle but marks his/her location at the blue triangle on an accurate map. In this case, we call the ground truth error as \textbf{\em marking error}. For example, this can be due to the lack of landmarks for the user to determine his/her exact location accurately on an open space on the map of the small screen of a mobile device, or simply due to a mistake from the user.

In Figure \ref{fig:errors}(b), the user would mark his/her location at the blue triangle if an accurate map is used. However, he/she actually marks location at the orange diamond due to using an incorrect map (as in Figure \ref{fig:japan}). In this case, we call this second ground truth error as \textbf{\em map error}.

\subsection{System Model and Problem Statement}

Next, consider a localization algorithm which returns a user's 2D location $X^{algo}$. Let $X^{gt}$ be the user's actual (real) ground truth location and $\widehat{X}^{gt}$ be the human-marked ground truth location.
Ideally we would like to calculate the algorithm's \textbf{\em real error}
$Err^{real} = X^{gt} - X^{algo}$.
However, due to the ground truth error
$Err^{gt} = \widehat{X}^{gt} - X^{gt}$,
the \textbf{\em validation error} actually used to evaluate the algorithm's performance in literature and many practical systems is
\begin{equation*}
Err^{val} = \widehat{X}^{gt} - X^{algo} = Err^{real} + Err^{gt}
\end{equation*}

Clearly, we are more interested in the location accuracy $|Err^{real}|$, rather than $|Err^{val}|$ which is impacted by the ground truth error. Once there is enough evaluation data ($\{X^{algo}, \widehat{X}^{gt}\}$), we can know the cumulative distribution function of $|Err^{val}|$. Furthermore, instead of using the full distribution, the accuracy is commonly expressed using high-level summary statistics, e.g. ``some inexpensive GPS receivers can locate positions to within 10 meters for approximately 95 percent of measurements \cite{Hightower2001}.'' In this paper, we consider three major summary statistics that are commonly used in literature and practical systems to quantify a localization algorithm performance:
\begin{itemize}
    \item The mean error $|Err^{val}|_{mean}$
    \item The 50\%-median error $|Err^{val}|_{median}$
    \item The 95\%-tail error $|Err^{val}|_{tail}$ \footnote{We select 95\% as it is most commonly used for confidence intervals \cite{Zar2007}. But in theory any level can be selected for the tail error such as 90\% or 99\%.}
\end{itemize}
The mean and median reflects the average user experience, while the tail reflects the worst user experience.

We again emphasize that, since we are more interested to know $|Err^{real}|_{mean}$ instead of $|Err^{val}|_{mean}$, we define the {\em \textbf{impact} of ground truth error}
\begin{equation}
\label{eq:delta_notion}
\Delta_{mean} = |Err^{val}|_{mean} - |Err^{real}|_{mean}
\end{equation}
as the difference between the validation error and the real error, i.e. how much of the validation error is caused by the ground truth error. The notions of $\Delta_{median}$ / $\Delta_{tail}$ are similar to (\ref{eq:delta_notion}).

\begin{figure}[t]
\vspace{-1mm}
\centerline
{\includegraphics[width=0.5\textwidth]{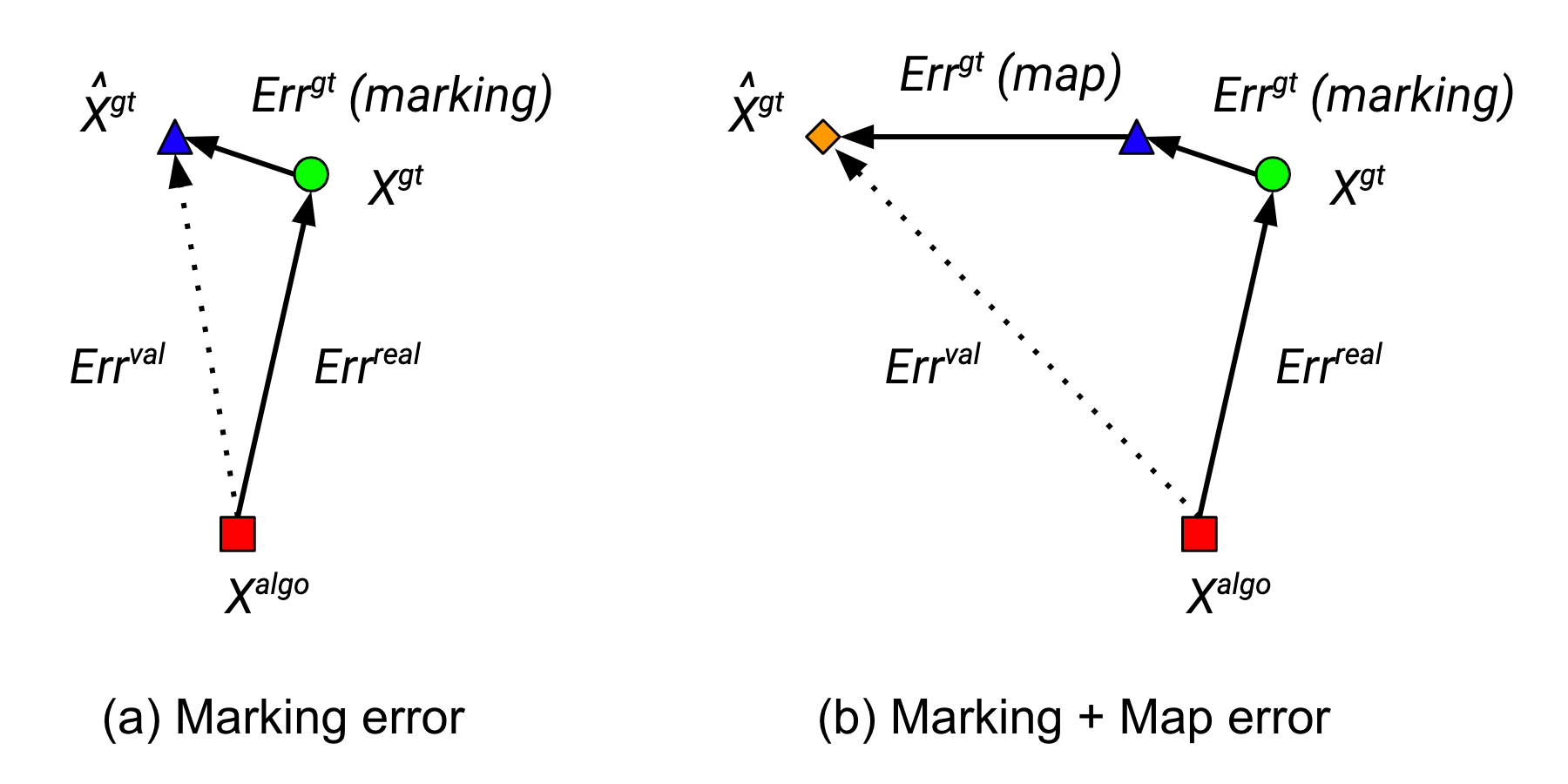}}
\caption{Ground truth errors. (a) A marking error between actual ground truth location (green circle) and human-marked ground truth location on an accurate map (blue triangle). (b) An additional map error between human-marked ground truth location on an accurate map (blue triangle) and that on an incorrect map (orange diamond).}
\label{fig:errors}
\end{figure}

To summarize, in this paper we propose and answer the following questions:
\begin{enumerate}
    \item Given the summary statistics on the validation error $|Err^{val}|_{mean}$ and the ground truth (marking/map) error $|Err^{gt}|_{mean}$, how can one computes the real algorithmic error $|Err^{real}|_{mean}$?
    \item What is the relation between the impact $\Delta_{mean}$ and the ground truth (marking/map) error $|Err^{gt}|_{mean}$?
    \item Can we answer the above two questions for other metrics such as the median or 95\%-tail error, or more generally for any $q$-th quantile?
    \item Which ground truth error has a larger impact on the evaluation of a localization system? Marking error or map error?
\end{enumerate}

\subsection{Application Example}
\label{sec:application}
Given a localization algorithm with validation error $|Err^{val}|_{mean} = 6m$ (which is the main metric commonly used in literature). The human-marked ground truth is collected on an incorrect map with an offset $|Err^{map}| = 2m$, and we estimate there is an average marking error $|Err^{mark}|_{mean} = 3m$.

First, we use Algorithm~\ref{algo:map} (Section \ref{sec:map_error}) to remove the map error, which gives a new validation error $|Err^{val}|_{mean} = 5.79m$ on an accurate map.
Next, we use Algorithm~\ref{algo:marking} (Section \ref{sec:marking_error}) to remove the marking error and end up with a real error $|Err^{real}|_{mean} = 4.95m$. This final error is a more realistic/accurate reflection of the algorithm performance in real time after removing the inaccuracy in ground truth collection.

\section{Theoretical Results}
\label{sec:theory}
\newtheorem{theorem}{Theorem}

Mathematically, instead of giving the summary statistics $|Err^{val}|_{mean}$ and $|Err^{gt}|_{mean}$, even if we know the full distribution of $Err^{val}$ and $Err^{gt}$, it is still hard to solve the distribution of $Err^{real}$ such that
$Err^{val} = Err^{real} + Err^{gt}$.
This is because the real algorithmic error and the ground truth error are independent, but the validation error and the ground truth error are not; and we do not know the correlation between them.

To make the problem mathematically tractable, we assume that the algorithmic and marking errors follow the most widely used normal distributions. We also further consider the map error of a constant translation. %

We start by analyzing the marking error in Section~\ref{sec:marking_error} and map error in Section~\ref{sec:map_error}. We then compare these two errors in Section~\ref{sec:compare_error}, and discuss the \textit{map scale error} in Section~\ref{sec:scale_error}. All assumptions are validated in Section~\ref{sec:evaluation}.

\subsection{Marking Error}
\label{sec:marking_error}

\begin{algorithm}[t]
\caption{Compute the real error from validation error and marking error. This algorithm uses mean but can be extended to any $q$-th quantile ($0 < q < 1$).}
\label{algo:marking}
\begin{algorithmic}
\renewcommand{\algorithmicrequire}{\textbf{Input:}}
\renewcommand{\algorithmicensure}{\textbf{Output:}}
\REQUIRE validation error mean $u = |Err^{val}|_{mean}$, marking error mean $v = |Err^{mark}|_{mean}$, $u > v$
\ENSURE real error mean $|Err^{real}|_{mean}$
\RETURN $|Err^{real}|_{mean} = \sqrt{u ^ 2 - v ^ 2}$
\end{algorithmic}
\end{algorithm}

Assume that both the real algorithmic error and the marking error follow 2D Gaussian distributions\footnote{We write $diag[{\sigma} ^ 2]$ for short to represent a matrix that has ${\sigma} ^ 2$ on all diagonal entries and zero otherwise.}
$Err^{real} \sim \mathcal{N}(0, diag[(\sigma^{real}) ^ 2])$,
$Err^{mark} \sim \mathcal{N}(0, diag[(\sigma^{mark}) ^ 2])$.
We also assume that both error distributions are symmetric in 2D (do not have any orientation bias), so that they have zero mean and the errors in X and Y directions are independent and identically distributed. In this setting, the validation error is also a 2D Gaussian
\begin{equation*}
Err^{val} = Err^{real} + Err^{mark} \sim \mathcal{N}(0, diag[(\sigma^{val})^2])
\end{equation*}
where
$(\sigma^{val}) ^ 2 = (\sigma ^ {real}) ^ 2 + (\sigma^{mark}) ^ 2$.

For the 1D error norm, it is known that for a Gaussian variable $X \sim \mathcal{N}(0, diag[{\sigma}^2])$, its norm $|X|$ follows a distribution $|X| \sim Rayleigh(\sigma)$ \cite{Rayleigh1880} with a probability density function
$p(x | \sigma) = \frac{x}{\sigma^{2}} e^{-\frac{x^2}{2 \sigma^2}}$, $x \geq 0$
and its mean $|X|_{mean}$ is equal to
$Rayleigh(\sigma)_{mean} = \sigma \sqrt{\pi / 2}$.

Thus $|Err^{mark}|$, $|Err^{real}|$ and $|Err^{val}|$ are all Rayleigh distributions. Since the mean is linear in $\sigma$,
\begin{align}
|Err^{real}|_{mean} & = \sigma^{real} \sqrt{\pi / 2} = \sqrt{\left( (\sigma^{val}) ^ 2 - (\sigma^{mark}) ^ 2 \right) \pi / 2} \nonumber \\
& = \sqrt{|Err^{val}|_{mean} ^ 2 - |Err^{mark}|_{mean} ^ 2} \label{eq:real_mean}
\end{align}

Algorithm~\ref{algo:marking} shows how to obtain the real error mean given the validation and marking error means. It also applies to the median and 95\%-tail, because for any $q$-th quantile ($0 < q < 1$),
$Rayleigh(\sigma)_{q \text{-th}} = \sigma \sqrt{-2 \ln{(1 - q)}}$
is always linear in $\sigma$.

\begin{figure}[t]
\vspace{-2mm}
\centerline
{\includegraphics[width=0.25\textwidth]{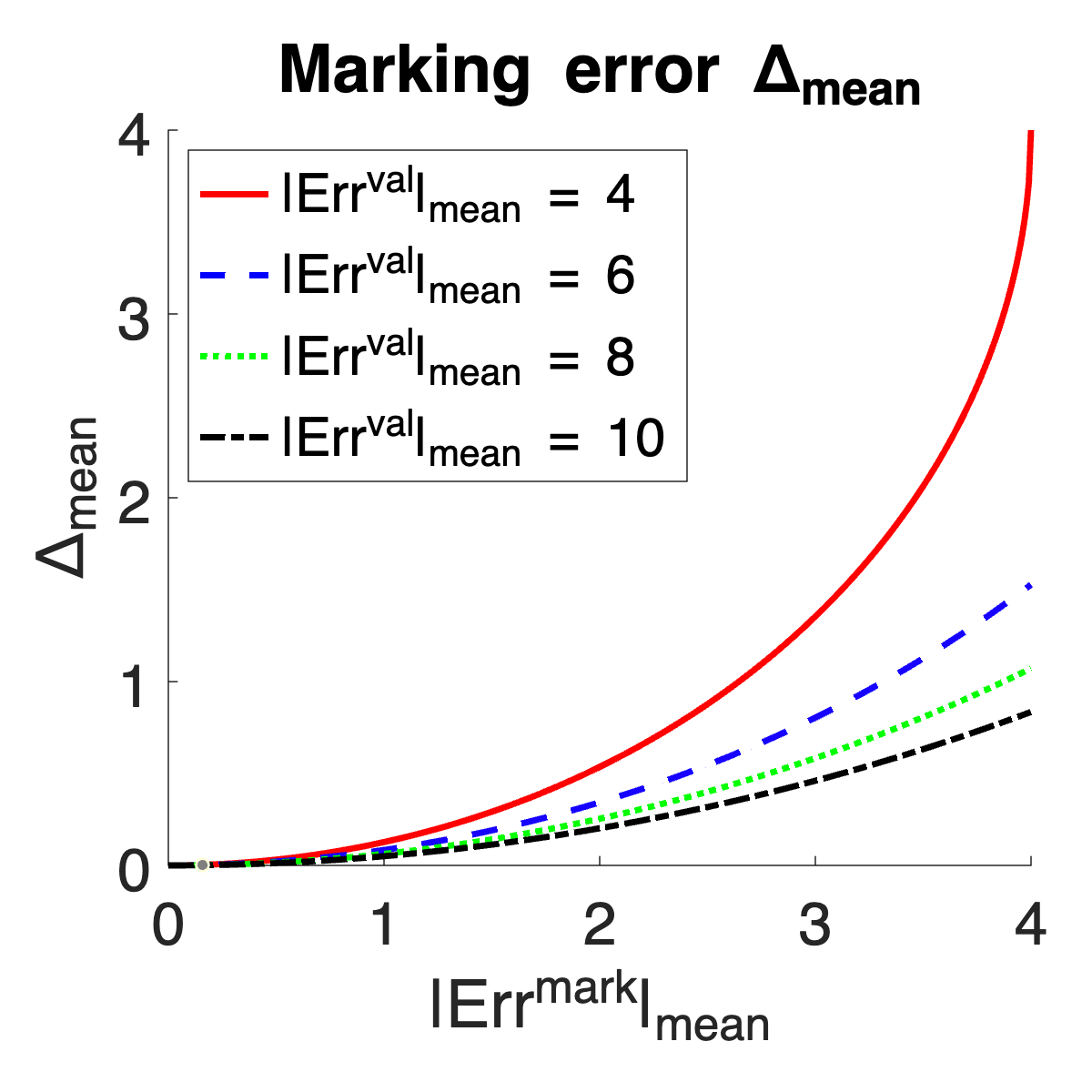}}
\caption{Impact of marking error $\Delta_{mean}$ as a function of $|Err^{val}|_{mean}$ and $|Err^{mark}|_{mean}$. The figures for $\Delta_{median}$ / $\Delta_{tail}$ are similar.}
\label{fig:impact_marking}
\vspace{-2mm}
\end{figure}

\begin{theorem}
\label{thm:impact_marking}
The impact of marking error ($\Delta_{mean}$):
\begin{equation*}
\frac{|Err^{mark}|_{mean}^2}{2 |Err^{val}|_{mean}} < \Delta_{mean}
< \frac{|Err^{mark}|_{mean}^2}{2 |Err^{val}|_{mean} - |Err^{mark}|_{mean}}
\end{equation*}
This also holds if we replace the mean by any $q$-th quantile (e.g. $\Delta_{median}$ and $\Delta_{tail}$).
\end{theorem}

\begin{proof}
From (\ref{eq:real_mean}),
$|Err^{val}|_{mean} - |Err^{mark}|_{mean} < |Err^{real}|_{mean} < |Err^{val}|_{mean}$.
\begin{align*}
\Delta_{mean} & = |Err^{val}|_{mean} - |Err^{real}|_{mean} \\
& = \frac{|Err^{val}|_{mean}^2 - |Err^{real}|_{mean}^2}{|Err^{val}|_{mean} + |Err^{real}|_{mean}} \\
& = \frac{|Err^{mark}|_{mean}^2}{|Err^{val}|_{mean} + |Err^{real}|_{mean}}
\end{align*}
\begin{equation*}
\frac{|Err^{mark}|_{mean}^2}{2 |Err^{val}|_{mean}} < \Delta_{mean}
< \frac{|Err^{mark}|_{mean}^2}{2 |Err^{val}|_{mean} - |Err^{mark}|_{mean}}
\end{equation*}
\end{proof}

Theorem \ref{thm:impact_marking} gives us a direct quantitative measure on the impact of marking error on location accuracy (see Figure~\ref{fig:impact_marking}):
\begin{enumerate}
    \item $\Delta$ is quadratic in $|Err^{mark}|$. This means that a small error in the collected ground truth has a magnified (quadratic) impact on the error in the reported algorithm accuracy. %
    \item $\Delta$ is inversely proportional to $|Err^{val}|$, i.e. when a localization algorithm does not perform well on validation, the impact of marking error is not significant. As the algorithm improves to have better accuracy, the quality of ground truth data becomes more and more important.
\end{enumerate}

\begin{theorem}
\label{thm:tail_median}
Relations between the real error mean, median, and 95\%-tail.
\begin{equation*}
\frac{|Err^{real}|_{tail}}{|Err^{real}|_{median}} = 2.07 ~ \textrm{and} ~
\frac{|Err^{real}|_{tail}}{|Err^{real}|_{mean}} = 1.95
\end{equation*}
These also hold if we replace $|Err^{real}|$ by $|Err^{val}|$.
\end{theorem}

\begin{proof}
\begin{align*}
\frac{|Err^{real}|_{tail}}{|Err^{real}|_{median}} & = \frac{\sigma^{real} \sqrt{-2 \ln{(1 - 0.95)}}}{\sigma^{real} \sqrt{-2 \ln{(1 - 0.5)}}} = 2.07 \\
\frac{|Err^{real}|_{tail}}{|Err^{real}|_{mean}} & = \frac{\sigma^{real} \sqrt{-2 \ln{(1 - 0.95)}}}{\sigma^{real} \sqrt{\pi / 2}} = 1.95
\end{align*}
\end{proof}

Theorem \ref{thm:tail_median} gives us an approximation on the average/worst user experience on location accuracy, e.g. if a GPS receiver has an accuracy of 10 meters for 95\% of the time, we would expect that it can localize within 5 meters for 50\% of the time.

\subsection{Map Error (Translation)}
\label{sec:map_error}

\begin{figure}[t]
\vspace{-1mm}
\centerline
{\includegraphics[width=0.25\textwidth]{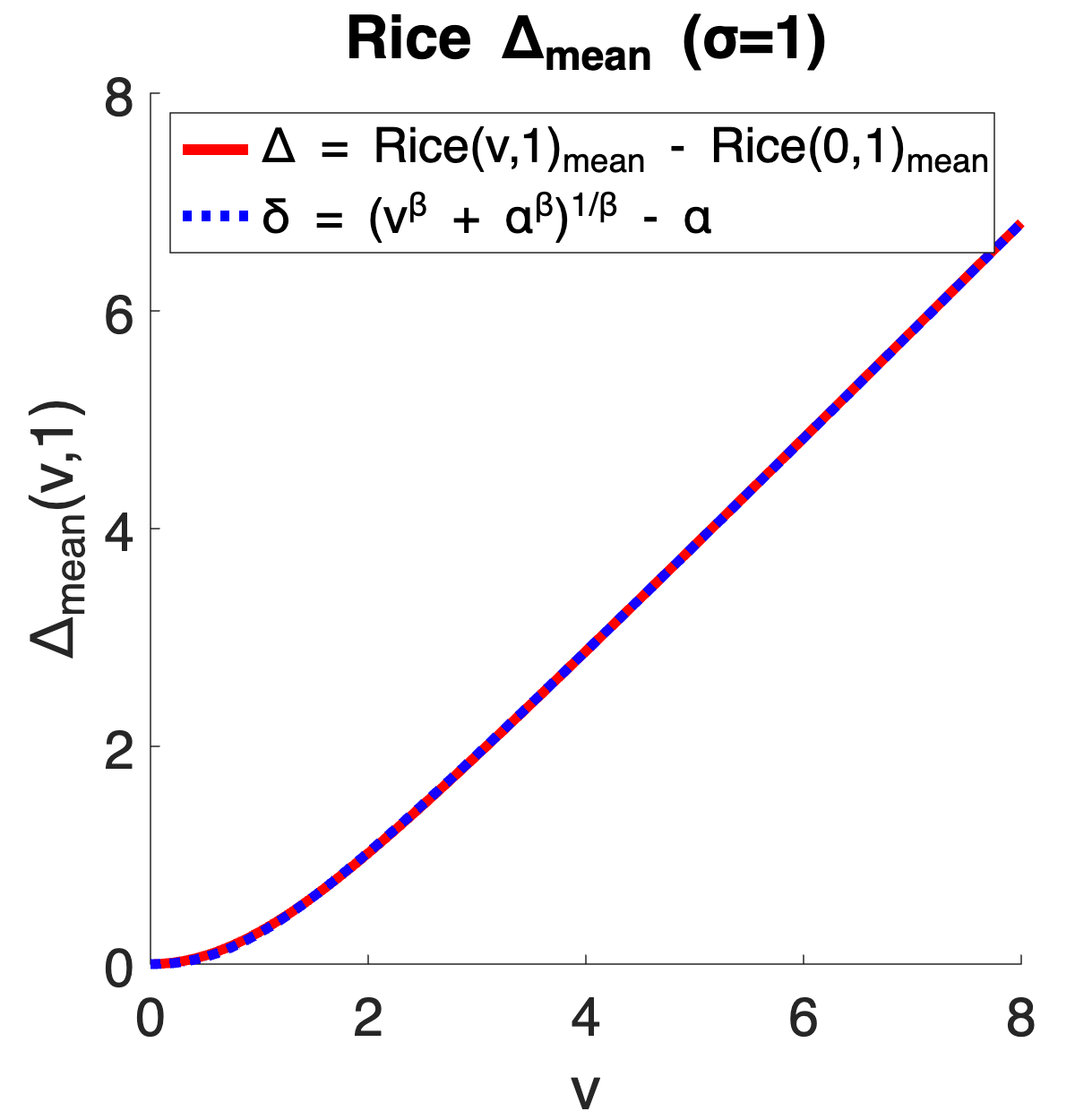}}
\caption{$\Delta_{mean}(v, 1)$ approximation function. $\Delta(v, 1)$ values are calculated by using the numerical Bessel function \cite{Schwartz2012}. $\delta(v, 1)$ values are calculated by using the algebraic approximation function. The figures for $\Delta_{median}$ / $\Delta_{tail}$ are similar.}
\label{fig:delta}
\end{figure}

Again, we assume that %
$Err^{real} \sim \mathcal{N}(0, diag[(\sigma^{real})^2])$
but
$Err^{map} = V$
is a constant 2D vector.
For now we assume that there is no marking error (see Section \ref{sec:application} when both marking and map errors exist). In this setting, the validation error is also a 2D Gaussian but with a non-zero mean
\begin{equation*}
Err^{val} = Err^{real} + Err^{map} \sim \mathcal{N}(V, diag[(\sigma^{real})^2])
\end{equation*}

For simplicity, let $\sigma = \sigma^{real}$ since there is only one variance parameter in this section (as the map error is a constant), and $v = |V|$ be the norm of the map error vector. The 1D error norm
$|Err^{real}| \sim Rayleigh(\sigma)$, $|Err^{map}| = v$.

For a non-zero mean Gaussian variable $X \sim \mathcal{N}(V, diag[{\sigma}^2])$, its norm $|X|$ follows a distribution $|X| \sim Rice(v, \sigma)$ \cite{Rice1944} with a probability density function
$p(x | v, \sigma) = \frac{x}{\sigma ^ 2} e ^ {- \frac{x ^ 2 + v ^ 2}{2 \sigma ^ 2}} I_0 \left( \frac{xv}{\sigma ^ 2} \right)$, $x \geq 0$
where
$I_0(z) = \sum_{k = 0}^{\infty} \frac{(z / 2) ^ {2k}}{(k!) ^ 2}$
is the modified Bessel function of the first kind with order zero.
Note that when $v = 0$, $Rice(v, \sigma)$ is the same as $Rayleigh(\sigma)$.

Finally, the impact of map error is
\begin{align*}
\Delta_{mean}(v, \sigma) & = |Err^{val}|_{mean} - |Err^{real}|_{mean} \\
& = Rice(v, \sigma)_{mean} - Rice(0, \sigma)_{mean}
\end{align*}
and similarly for $\Delta_{median}$ / $\Delta_{tail}$.

\subsubsection{$\Delta(v, \sigma)$ Approximation Function}

\begin{table}[t]
\caption{Constants $\alpha$, $\beta$ for $\delta(v, 1)$ approximation function and root-mean-square error on $|\Delta(v, 1) - \delta(v, 1)|$.}
\begin{center}
\vspace{-1mm}
\begin{tabular}{|c|c|c|c|}
\hline
 & $\alpha$ & $\beta$ & $RMSE~|\Delta - \delta|$ \\
\hline
$\text{Mean}$ & 1.2392 & 2.3064 & 0.0052 \\
$\text{Median}$ & 1.1471 & 2.3384 & 0.0032 \\
$\text{Tail}$ & 0.7870 & 1.9452 & 0.0038 \\
\hline
\end{tabular}
\label{table:alpha_beta}
\end{center}
\end{table}

\begin{table}[t]
\caption{Constant $\gamma$ for $Rice(0, 1)/Rayleigh(1)$.}
\vspace{-1mm}
\begin{center}
\begin{tabular}{|c|c|}
\hline
 & $\gamma$ \\
\hline
$\text{Mean}$ & $\sqrt{\pi / 2} = 1.2533$ \\
$\text{Median}$ & $\sqrt{-2 \ln{(1 - 0.5)}} = 1.1774$ \\
$\text{Tail}$ & $\sqrt{-2 \ln{(1 - 0.95)}} = 2.4477$ \\
\hline
\end{tabular}
\label{table:gamma}
\end{center}
\vspace{-1mm}
\end{table}

Unlike Rayleigh distribution, the mean/quantile statistics of Rice distribution have no closed-form expressions. This is because in general, Bessel functions cannot be expressed as a finite algebraic combination of elementary functions \cite{Ritger1968}. To make the analysis easier, we give an algebraic approximation function for computing the impact $\Delta(v, \sigma)$.

We start with analysis on the mean error. By scaling a factor of $1/\sigma$, we have
\begin{align}
\Delta(v, \sigma) & = Rice(v, \sigma)_{mean} - Rice(0, \sigma)_{mean} \nonumber \\
& = \sigma \left( Rice(v / \sigma, 1)_{mean} - Rice(0, 1)_{mean} \right) \nonumber \\
& = \sigma \Delta (v / \sigma, 1) \label{eq:delta_scale}
\end{align}

Consider the function
\begin{equation*}
\Delta(v, 1) = Rice(v, 1)_{mean} - Rice(0, 1)_{mean}
\end{equation*}
where
$Rice(0, 1)_{mean} = Rayleigh(1)_{mean} = \sqrt{\pi / 2}$.

It has the following two properties:
\begin{itemize}
    \item $\Delta(0, 1) = 0$
    \item $\lim_{v \to \infty} \frac{\Delta(v, 1)}{v} = 1$ \footnote{As $v \to \infty$, if v increases by 1, the norm $|X|$ of every sample $X \sim Rice(v, 1)$ increases by an amount that also converges to 1. Thus $\Delta(v, 1)$ is linear in $v$ with slope 1.}
\end{itemize}

We next construct an algebraic approximation function
\begin{equation}
\label{eq:approx_func}
\delta(v, 1) = (v ^ \beta + \alpha ^ \beta) ^ \frac{1}{\beta} - \alpha
\end{equation}
with two rational number constants $\alpha > 0$, $\beta > 1$. Here $\alpha$ sets the difference between $v$ and $\delta(v, 1)$ as $v \to \infty$, and $\beta$ controls how $\delta(v, 1)$ varies for small $v$.
One can verify that $\delta(v, 1)$ has two similar properties:
\begin{itemize}
    \item $\delta(0, 1) = 0$
    \item $\lim_{v \to \infty} \frac{\delta(v, 1)}{v} = 1$
\end{itemize}

Figure \ref{fig:delta} shows a comparison of $\Delta(v, 1)$ and its approximation $\delta(v, 1)$. The $\Delta(v, 1)$ values are computed by generating the probability density function using numerical calculations of the Bessel function \cite{Schwartz2012}, and the approximation $\delta(v, 1)$ values are computed from the closed-form expression (\ref{eq:approx_func}). The constants $\alpha$ and $\beta$ are optimized to minimize the root-mean-square error on $|\Delta(v, 1) - \delta(v, 1)|$ (see Table \ref{table:alpha_beta}).

Once we have $\delta(v, 1)$, similar to (\ref{eq:delta_scale})
\begin{equation*}
\delta(v, \sigma) = \sigma \delta (v / \sigma, 1) = \sigma \left( \left( (v / \sigma) ^ \beta + \alpha ^ \beta \right) ^ {\frac{1}{\beta}} - \alpha \right)
\end{equation*}

\subsubsection{Numerical Algorithm for Computing the Real Error (Algorithm \ref{algo:map})}

\begin{algorithm}[t]
\caption{Compute the real error from validation error and map error. This algorithm uses mean but can be extended to any $q$-th quantile ($0 < q < 1$).}
\label{algo:map}
\begin{algorithmic}
\renewcommand{\algorithmicrequire}{\textbf{Input:}}
\renewcommand{\algorithmicensure}{\textbf{Output:}}
\REQUIRE validation error mean $u = |Err^{val}|_{mean}$, map error $v = |Err^{map}|$, $u > v$
\ENSURE real error mean $|Err^{real}|_{mean}$
  \STATE Find $\alpha, \beta$ for mean error in Table \ref{table:alpha_beta}, and $\gamma$ in Table \ref{table:gamma}
  \STATE Set $\sigma_{min} = (u - v) / \gamma$, $\sigma_{max} = (u - v) / (\gamma - \alpha)$
  \STATE Set $\epsilon$ for convergence
  \WHILE {$\sigma_{max} - \sigma_{min} > \epsilon$}
    \STATE $\sigma = (\sigma_{min} + \sigma_{max}) / 2$
    \IF {$(u / \sigma + \alpha - \gamma) ^ \beta - (v / \sigma) ^ \beta - \alpha ^ \beta < 0$}
      \STATE $\sigma_{max} = \sigma$
    \ELSE
      \STATE $\sigma_{min} = \sigma$
    \ENDIF
  \ENDWHILE
\RETURN $|Err^{real}|_{mean} = \sigma \gamma$
\end{algorithmic}
\end{algorithm}

Using the $\Delta(v, \sigma)$ approximation function, we can obtain the real error mean as follows:

Let $u = |Err^{val}|_{mean}$, $\gamma = Rice(0, 1)_{mean} = \sqrt{\pi / 2}$.
By definition of $\Delta_{mean}$,
\begin{align*}
|Err^{val}|_{mean} & = |Err^{real}|_{mean} + \Delta(v, \sigma) \\
u & = \sigma \gamma + \sigma \left( \left( (v / \sigma) ^ \beta + \alpha ^ \beta \right) ^ \frac{1}{\beta} - \alpha \right)
\end{align*}

This is equivalent to solving the equation
\begin{equation}
\label{eq:solve_sigma}
f(\sigma) = (u / \sigma + \alpha - \gamma) ^ \beta - (v / \sigma) ^ \beta - \alpha ^ \beta = 0
\end{equation}

Let $\sigma_{min} = \frac{u - v}{\gamma}$, $\sigma_{max} = \frac{u - v}{\gamma - \alpha}$.
\begin{align*}
f(\sigma_{min}) & = (v / \sigma_{min} + \alpha) ^ \beta - (v / \sigma_{min}) ^ \beta - \alpha ^ \beta > 0 \\
f(\sigma_{max}) & = - \alpha ^ \beta < 0
\end{align*}

and $\forall{\sigma < \sigma_{max}}$,
\begin{equation*}
f'(\sigma) = - \beta \left( (u / \sigma + \alpha - \gamma) ^ {\beta - 1} - (v / \sigma) ^ {\beta - 1} \right) \Big/ {\sigma^2} < 0
\end{equation*}

Therefore, we can binary search on $\sigma \in (\sigma_{min}, \sigma_{max})$ to solve $f(\sigma) = 0$. Finally,
$|Err^{real}|_{mean} = \sigma \gamma$.

Figure \ref{fig:impact_map} shows the impact of map error $\Delta(v, \sigma)$, which seems to increase slower (comparing to marking error) until when the map error $v$ is close to the validation error $u$. Intuitively, this can be explained as follows:
Let the real error be a zero mean Gaussian centered at the origin $O$, and we add a map error $V = (v, 0)$, where $v > 0$.
For any sample $X$, its error increases $|X - V| > |X - O|$ if and only if $X$ lies on the left side of the line $x = v / 2$.
Since the Gaussian has zero mean, there are more samples on the left side of the line $x = v / 2$, so the overall mean error will increase. But it increases slowly as the changes on two sides may cancel each other, until for sufficiently large $v$ where a significant majority of samples are on the left side of the line $x = v / 2$.

Similarly, Algorithm~\ref{algo:map} can be applied to any $q$-th quantile ($0 < q < 1$). Table~\ref{table:alpha_beta} and \ref{table:gamma} also list $\alpha$, $\beta$, and $\gamma$ values for the median and 95\%-tail.

Note that unlike marking error, for map error, Theorem~\ref{thm:tail_median} only holds for the real error
$\frac{|Err^{real}|_{tail}}{|Err^{real}|_{median}} > 2$,
but not for the validation error as
$\lim_{v \to \infty} \frac{|Err^{val}|_{tail}}{|Err^{val}|_{median}} = \lim_{v \to \infty} \frac{|Err^{real}|_{tail} + \Delta_{tail}(v, \sigma)}{|Err^{real}|_{median} + \Delta_{median}(v, \sigma)} = 1$, because $\lim_{v \to \infty} \frac{\Delta(v, \sigma)}{v} = 1$.

\begin{figure}[t]
\centerline
{\includegraphics[width=0.25\textwidth]{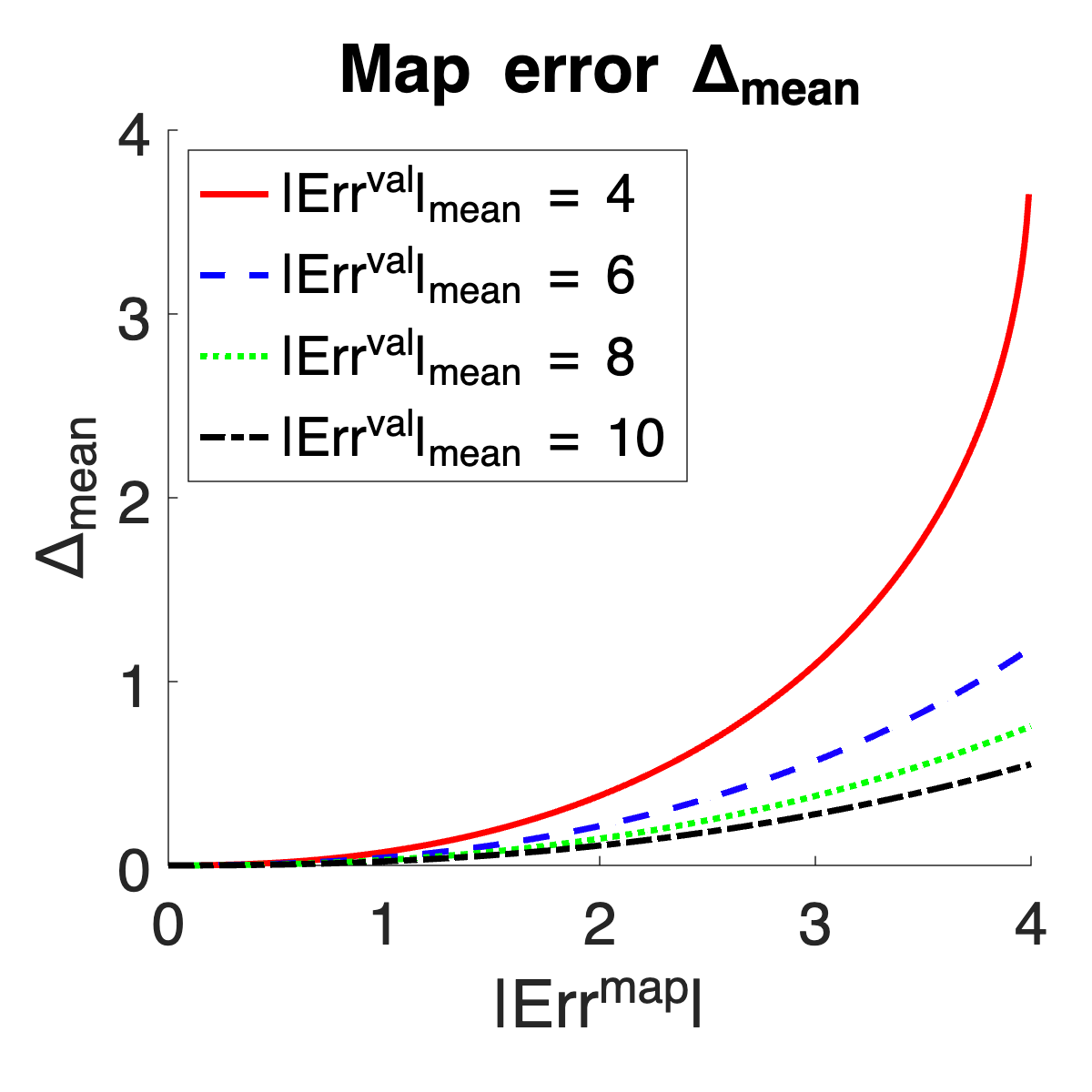}}
\caption{Impact of map error $\Delta_{mean}$ as a function of $|Err^{val}|_{mean}$ and $|Err^{map}|$. These curves are similar to Fig.~\ref{fig:impact_marking} on the marking error, but values are smaller.}
\label{fig:impact_map}
\end{figure}

\subsection{Comparing the Impact of Marking and Map Errors}
\label{sec:compare_error}

In Section~\ref{sec:map_error}, we briefly explained that the impact of map error is less than marking error on $\Delta_{mean}$, when the map error is equal to the marking error mean. Now we formally prove this result.

\begin{theorem}
\label{thm:impact_compare}
Given the validation error $u = |Err^{val}|_{mean}$ and the ground truth error $v = |Err^{gt}|_{mean}$ (either marking or map error), $u > v$. Let $\Delta_{mean}^{mark}$ be the impact of marking error when $v$ is $|Err^{mark}|_{mean}$, and $\Delta_{mean}^{map}$ be the impact of map error when $v$ is $|Err^{map}|$. With similar definitions on $\Delta_{median}$ / $\Delta_{tail}$, we have
\begin{align*}
& \Delta_{mean}^{mark} > \Delta_{mean}^{map}, ~ \forall{v / u < 0.9995} \\
& \Delta_{median}^{mark} > \Delta_{median}^{map}, ~ \forall{v / u < 0.9970} \\
& \Delta_{tail}^{mark} < \Delta_{tail}^{map}
\end{align*}
\end{theorem}

\begin{proof}
We start with $\Delta_{mean}$. From Algorithm~\ref{algo:marking}, $\Delta_{mean}^{mark} = u - \sqrt{u ^ 2 - v ^ 2}$. From Algorithm \ref{algo:map}, $\Delta_{mean}^{map} = u - \sigma^{*} \gamma$ where $\sigma^{*} \in (\sigma_{min}, \sigma_{max})$ is the solution for $f(\sigma) = 0$ in (\ref{eq:solve_sigma}).

So $\Delta_{mean}^{mark} > \Delta_{mean}^{map}$ if and only if
\begin{equation}
\label{eq:proof_sigma}
\sigma^{*} > \sqrt{u ^ 2 - v ^ 2} \Big/ \gamma
\end{equation}

Let $s = \sqrt{\frac{u + v}{2}}$, $t = \sqrt{\frac{u - v}{2}}$, $\lambda = t / s < 1$.

Then $u = s^2 + t^2$, $v = s^2 - t^2$, $\frac{v}{u} = \frac{1 - \lambda^2}{1 + \lambda^2}$.

(a) When $\lambda = \frac{t}{s} \le \frac{\gamma - \alpha}{\gamma}$, i.e. $\frac{v}{u} \ge \frac{1 - \left( \frac{\gamma - \alpha}{\gamma} \right)^2}{1 + \left( \frac{\gamma - \alpha}{\gamma} \right)^2} = 0.9997$,
\begin{equation*}
\frac{\sqrt{u ^ 2 - v ^ 2}}{\gamma} = \frac{2st}{\gamma} \ge \frac{2t^2}{\gamma - \alpha} = \frac{u - v}{\gamma - \alpha} = \sigma_{max} > \sigma^{*}
\end{equation*}

In this case, $\Delta_{mean}^{mark} < \Delta_{mean}^{map}$.

(b) When $\lambda = \frac{t}{s} > \frac{\gamma - \alpha}{\gamma}$, i.e. $\frac{v}{u} < 0.9997$,
\begin{equation*}
\sqrt{u ^ 2 - v ^ 2} \Big/ \gamma < \sigma_{max}, ~ \sigma^{*} < \sigma_{max}
\end{equation*}

Since $\forall{\sigma < \sigma_{max}}$, $f'(\sigma) < 0$. (\ref{eq:proof_sigma}) is equivalent to
\begin{equation*}
f \left( \sqrt{u ^ 2 - v ^ 2} \Big/ \gamma \right) > f(\sigma^{*}) = 0 \end{equation*}
\begin{equation*}
\left( \frac{u}{\frac{\sqrt{u ^ 2 - v ^ 2}}{\gamma}} + \alpha - \gamma \right) ^ \beta - \left( \frac{v}{\frac{\sqrt{u ^ 2 - v ^ 2}}{\gamma}} \right) ^ \beta - \alpha ^ \beta > 0
\end{equation*}

Replace $(u, v)$ by $(s, t)$,
\begin{equation*}
\left( \frac{s^2 + t^2}{2st} + \frac{\alpha}{\gamma} - 1 \right) ^ \beta - \left( \frac{s^2 - t^2}{2st} \right) ^ \beta - \left( \frac{\alpha}{\gamma} \right) ^ \beta > 0
\end{equation*}

Using $\lambda = t / s$,
\begin{align}
g(\lambda) & = \left( 1 + 2 \lambda (\alpha / \gamma - 1) + \lambda^2 \right) ^ \beta - (1 - \lambda^2) ^ \beta - ( 2 \lambda \alpha / \gamma) ^ \beta \nonumber \\
& > 0 \label{eq:lambda}
\end{align}

Figure \ref{fig:lambda} plots the algebraic function $g(\lambda)$ and its derivative $g'(\lambda)$ on $0 \le \lambda \le 1$. They satisfy
\begin{align*}
& g(0) = 0, ~ g'(0) = 2 \beta ( \alpha / \gamma - 1 ) < 0 \\
& g(1 - \alpha / \gamma) = - ( 2 \alpha (\gamma - \alpha) / \gamma ^2 ) ^ \beta < 0 \\
& g(1) = 0, ~ g'(1) = 0
\end{align*}

From $g'(\lambda)$, we know there exists a root $g(\lambda^{*}) = 0$ such that $\forall{0 < \lambda < \lambda^{*}}$, $g(\lambda) < 0$ and $\forall{\lambda^{*} < \lambda < 1}$, $g(\lambda) > 0$. Therefore, we can binary search $g(\lambda) = 0$ for a numerical solution $\lambda^{*} = 0.0160$.
So in this case, (\ref{eq:lambda}) holds if and only if $\lambda > \lambda^{*}$, i.e.
$v / u < \frac{1 - {\lambda^{*}}^2}{1 + {\lambda^{*}}^2} = 0.9995$.

Combining (a) and (b), $\Delta_{mean}^{mark} > \Delta_{mean}^{map}$ if and only if $v / u < 0.9995$.

The proof for $\Delta_{median}$ is similar. $\lambda^{*} = 0.0390$ and $\Delta_{median}^{mark} > \Delta_{median}^{map}$ if and only if $v / u < 0.9970$.

$\Delta_{tail}$ is different. As shown in Figure \ref{fig:lambda}, $\forall{0 < \lambda < 1}$, $g(\lambda) < 0$. So (\ref{eq:lambda}) never holds and $\Delta_{tail}^{mark} < \Delta_{tail}^{map}$.
\end{proof}

Although theoretically there are examples $\Delta_{mean}^{mark} < \Delta_{mean}^{map}$ (e.g. $\lambda = \frac{\gamma - \alpha}{\gamma}$). In practice, $v / u > 0.99$ should never happen (ground truth error being more than 99\% of validation error). Thus, we conclude that marking error has a larger impact on the mean and median, while map error has a larger impact on the tail.

\begin{figure}[t]
\centerline
{\includegraphics[width=0.5\textwidth]{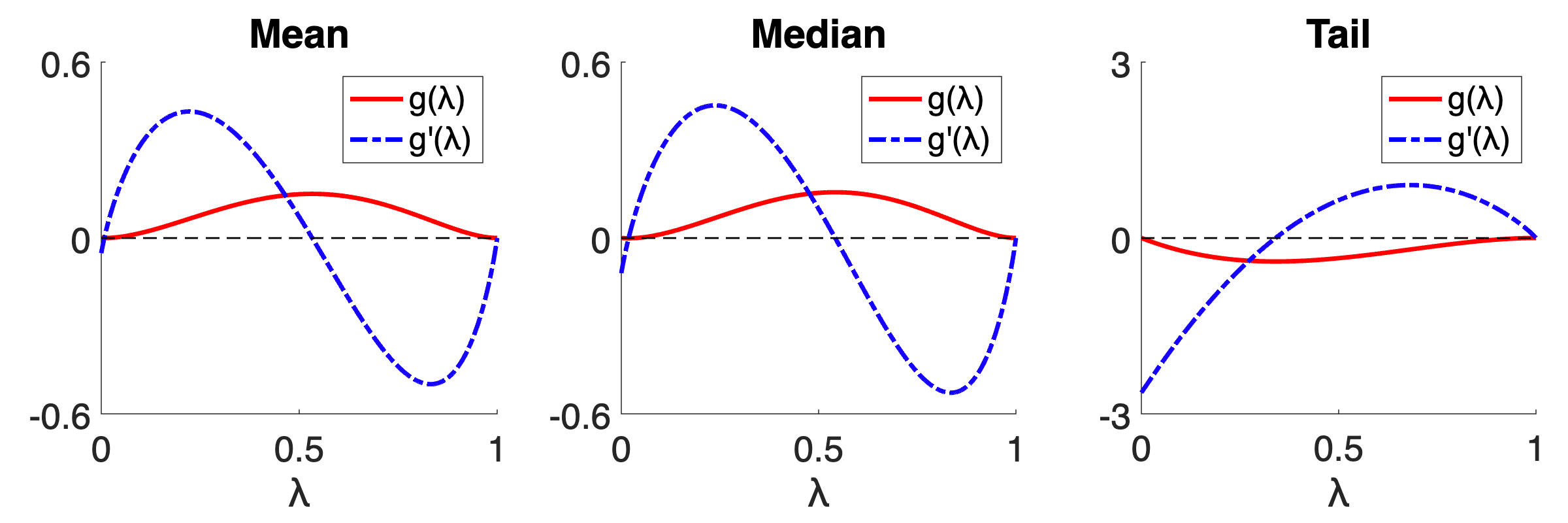}}
\caption{Algebraic function $g(\lambda)$ and its derivative $g'(\lambda)$.}
\label{fig:lambda}
\end{figure}

\subsection{Map Error (Scale)}
\label{sec:scale_error}

Besides the translation error, another type of map error is the scale error, which can be modeled by a linear transform
$\widehat{X}^{gt} = S X^{gt} + V$
where $S = diag[s_x, s_y]$ is a scaling matrix without rotation, and $V$ is a translation vector. Then
\begin{equation*}
Err^{map} = \widehat{X}^{gt} - X^{gt} = (S - I) X^{gt} + V
\end{equation*}
where $I$ is an identity matrix. For $S = I$, this reduces to the model we discussed in Section \ref{sec:map_error}.

Note that in this paper, our analysis only relies on the error distribution $Err^{gt}$. %
It does not require the distribution of ground truth location $X^{gt}$ (i.e. where the data is collected). With the scale error ($S \ne I$), the distribution of $Err^{map}$ depends on the distribution of $X^{gt}$. Unless $X^{gt}$ also follows a normal distribution and $s_x = s_y$, we cannot apply the same analysis in this paper. This additional distribution of $X^{gt}$ makes the model more complicated. We leave this as a topic for future work.

\section{Experimental Validation}
\label{sec:evaluation}

In this section, we validate the theoretical assumptions and results on a real dataset collected in a typical environment. We start by describing the environment and the localization technique. We then validate the error distributions assumptions. Finally, we compare our theoretical results with experimental results obtained directly from the real data.

\subsection{Experimental Setup}

Without loss of generality, we use a probabilistic WiFi fingerprinting localization system in our experiments similar to \cite{youssef2005horus, youssef2003wlan}. %
To collect the necessary data for validation, we deploy the localization system in a floor in our university campus building with a 37m x 17m area containing labs, offices, meeting rooms as well as corridors (Figure~\ref{fig:testbed}). We use the already installed WiFi infrastructure in the building, mainly four APs installed in the same floor in addition to 12 APs overheard from other floors. %
We obtain the real ground truth (without marking or map errors) based on the \textbf{\em{landmarks in the environment}} such as doors and fixed access points locations. We collect the WiFi scans by a Samsung S4 cell phone that scans for the WiFi access points at different 24 discrete reference locations that cover the entire area of interest uniformly.

\begin{figure}[t]
\centerline
{\includegraphics[width=0.5\textwidth]{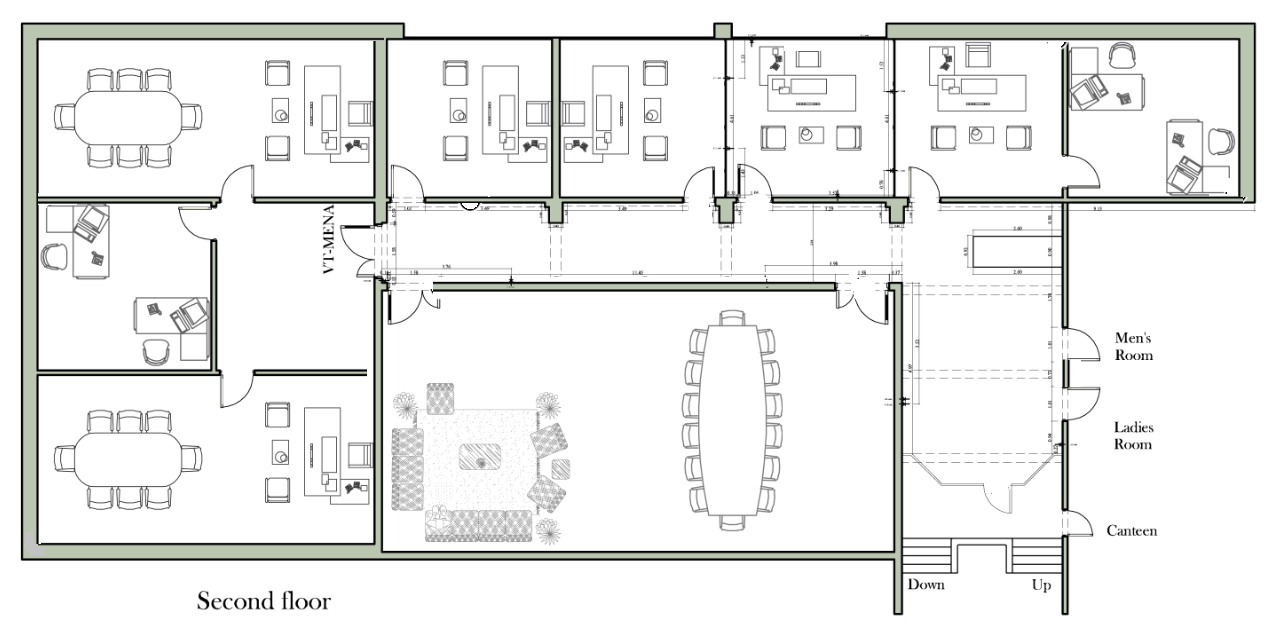}}
\caption{Indoor testbed used for experimental validation. }
\label{fig:testbed}
\end{figure}

\subsection{Validating Error Distributions Assumptions}

\subsubsection{Marking Error}

We start by getting the marking error distribution by comparing the marked ground truth locations, which are obtained by clicking on the map, with the real ground truth locations from the floor landmarks. %
Our results show that the error in X follows a normal distribution with $\mu = -0.04$ and $\sigma = 0.16$, i.e. $\mathcal{N}(-0.04, 0.16^2)$, while the error in Y follows a normal distribution $\mathcal{N}(-0.03, 0.09^2)$. The mean values in both distributions are small around zero which matches our theoretical assumptions that the marking errors in X and Y follow normal distributions with mean zero.
The quantile-quantile plot in Figure~\ref{fig:marking2} shows that the marking error norm fits a Rayleigh distribution with $\sigma = 0.14$ very well. This also matches our theoretical analysis. %

\subsubsection{Validation Error}

The validation error distribution is obtained by comparing the system estimated location with the marked (not real) ground truth locations.
Our results show that the error in X follows a normal distribution $\mathcal{N}(0.84, 3.86^2)$, while the error in Y follows a normal distribution $\mathcal{N}(0.41, 1.54^2)$. In our analysis, we assume that the error distributions in X and Y are identical. However, in our experiments, the validation error in X direction is higher than in Y direction. This is due to the geometry of the building (e.g. the building width is larger than its height, leading to more errors in the X direction), making the validation error norm possibly deviate from the Rayleigh distribution. %

Figure~\ref{fig:val2} shows the validation error norm distribution
and its best fit using a Rayleigh distribution with $\sigma = 2.81$. Since in our experiments the error distributions in X and Y are not identical, Rayleigh distribution does not lead to the best fit and other distributions, e.g. the exponential distribution in Figure~\ref{fig:val2}(c), could be a better choice. Nonetheless, Rayleigh distribution still fits most percentiles, except for very large tail errors.

\begin{figure}[!t]
        \centering
        \begin{subfigure}[b]{0.25\textwidth}
                \includegraphics[width=0.99\linewidth]{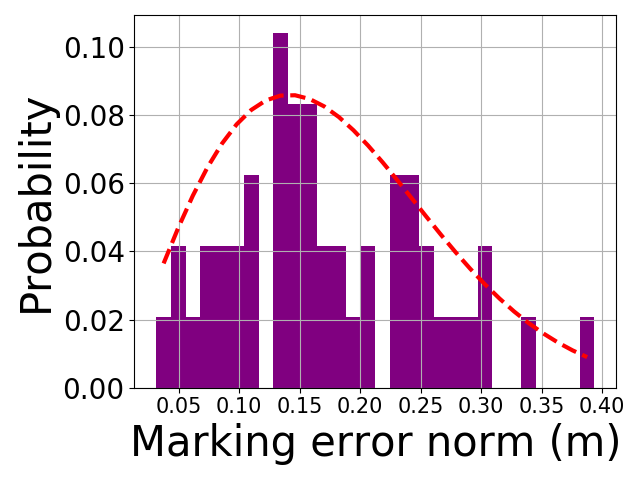}
           \caption{Marking error distribution.}
           \label{fig:marking2a}
        \end{subfigure}%
        \begin{subfigure}[b]{0.25\textwidth}
                \includegraphics[width=0.99\linewidth]{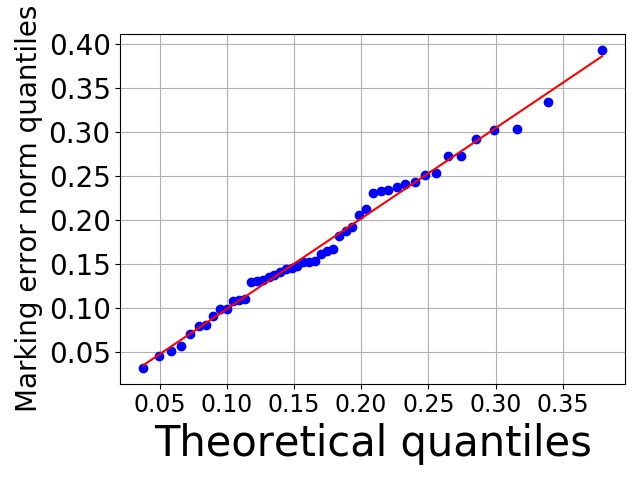}
              \caption{Rayleigh distribution test.}
            \label{fig:marking2b}
        \end{subfigure}
        \caption{Marking error norm distribution with Q-Q plot vs. Rayleigh distribution. %
        }
             \label{fig:marking2}
\end{figure}

\begin{figure*}[!t]
\vspace{-1mm}
        \centering
        \begin{subfigure}[b]{0.25\textwidth}
                \includegraphics[width=0.99\linewidth]{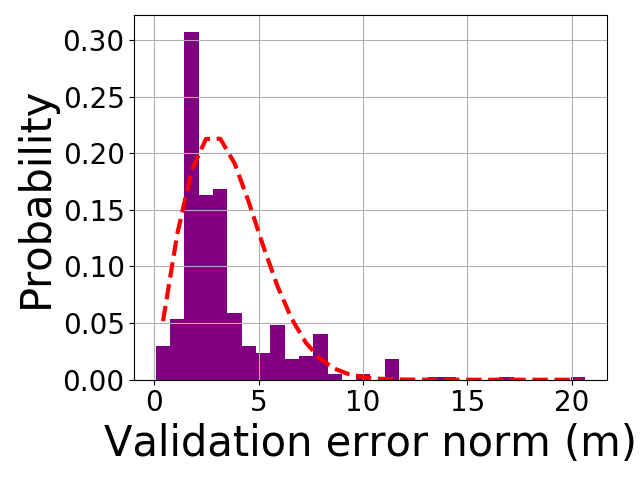}
           \caption{Validation error distribution.}
           \label{}
        \end{subfigure}%
        \begin{subfigure}[b]{0.25\textwidth}
                \includegraphics[width=0.99\linewidth]{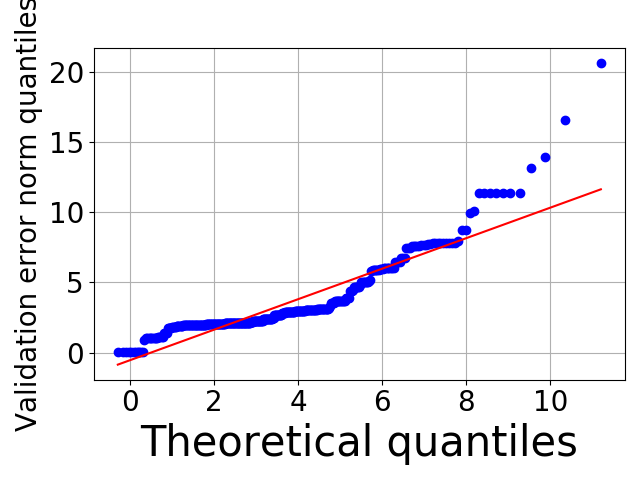}
           \caption{Rayleigh distribution test.}
           \label{}
        \end{subfigure}%
         \begin{subfigure}[b]{0.25\textwidth}
                \includegraphics[width=0.99\linewidth]{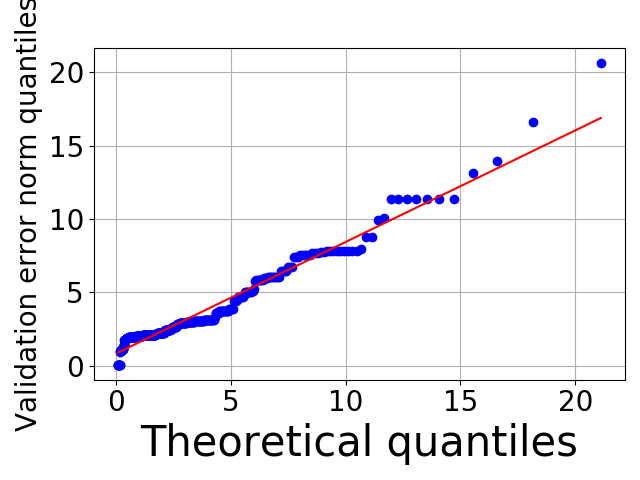}
              \caption{Exponential distribution test.}
            \label{}
        \end{subfigure}
        \caption{Validation error norm distribution with Q-Q plot vs. Rayleigh and exponential distributions. The figures for real error norm distribution are similar.}
             \label{fig:val2}
\vspace{-1mm}
\end{figure*}

\subsubsection{Real Error}

The real error distribution is obtained by comparing the system estimated location with the real ground truth locations in the collected dataset (obtained based on the landmarks in the environment).
Similar to the validation error, the real error distributions in X and Y follow normal distributions $\mathcal{N}(0.80, 3.89^2)$ and $\mathcal{N}(0.39, 1.53^2)$, respectively. The real error norm also fits a Rayleigh distribution with $\sigma = 2.78$ except at the tail. Similarly, this is due to the error distributions in X and Y being not identical. %

\subsubsection{Validation Map Error}

To get the validation map error distribution, we shift the real ground truth locations in the evaluation dataset in X and Y directions, simulating offset errors in the provided map. Figure~\ref{fig:map2} shows the validation error norm distribution after shifting the map by 3m in both X and Y directions.
The validation error norm fits the Rice distribution with $v = 4.65$ and $\sigma = 3.29$ %
better than other distributions most of the time, except at the tail (asymmetric X/Y distribution). Note that when the map error exists, although the exponential distribution gives larger tail errors, it does not perform well at the lower percentiles.

\begin{figure*}[!t]
\vspace{-1mm}
        \centering
        \begin{subfigure}[b]{0.25\textwidth}
                \includegraphics[width=0.99\linewidth]{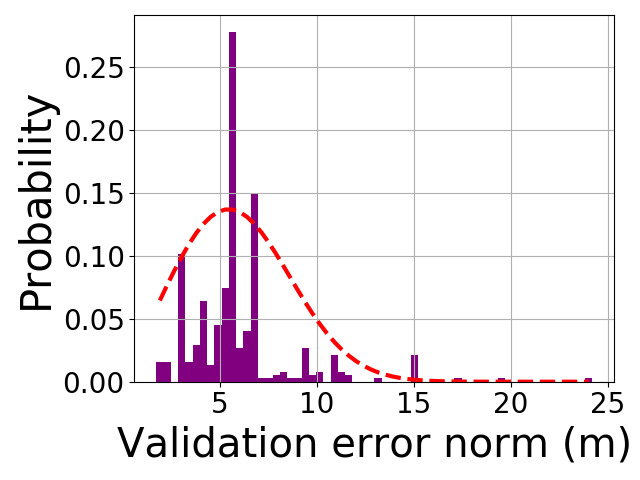}
           \caption{Validation error distribution.}
           \label{}
        \end{subfigure}%
        \begin{subfigure}[b]{0.25\textwidth}
                \includegraphics[width=0.99\linewidth]{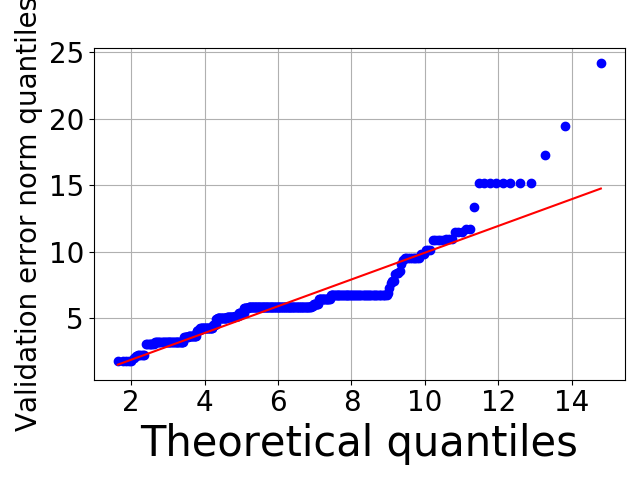}
           \caption{Rice distribution test.}
           \label{}
        \end{subfigure}%
                 \begin{subfigure}[b]{0.25\textwidth}
                \includegraphics[width=0.99\linewidth]{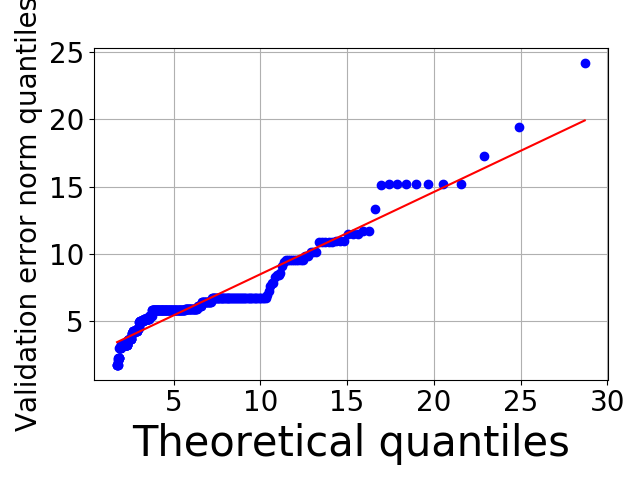}
              \caption{Exponential distribution test.}
            \label{}
        \end{subfigure}
        \caption{Validation error norm distribution after shifting the map by 3m in both X and Y directions with Q-Q plot vs. Rice and exponential distributions.}
             \label{fig:map2}
\vspace{-1mm}
\end{figure*}

\subsection{Theoretical vs Experimental Real Error}

\subsubsection{Marking Error}

Figure~\ref{fig:cdf_exp_thr} presents a CDF comparison between the theoretical and experimental real errors. The theoretical real error is obtained from the marking and validation errors by applying Algorithm~\ref{algo:marking}. The experimental real error is obtained directly from the data. The figure shows that the theoretical real error matches the experimental one within 4\% in all percentiles.
Table~\ref{table:exp_th_err} further shows the summary statistics of the results.

On Theorem~\ref{thm:tail_median}, the ratio of the validation error tail to the median is $\frac{8.05}{2.47} = 3.26$, and the ratio of the validation error tail to the mean is $\frac{8.05}{3.41} = 2.36$.
Similarly, the ratio of the real error tail to the median is $\frac{7.99}{2.37} = 3.37$ and the ratio of the real error tail to the mean is $\frac{7.99}{3.39} = 2.36$.
These show that in practice, the 95\%-tail error is more than twice of the median/mean errors. Again, these higher ratios are because of the larger tail errors from asymmetric X/Y distribution.

\begin{figure}[t]
\vspace{-1mm}
\centerline
{\includegraphics[width=0.3\textwidth]{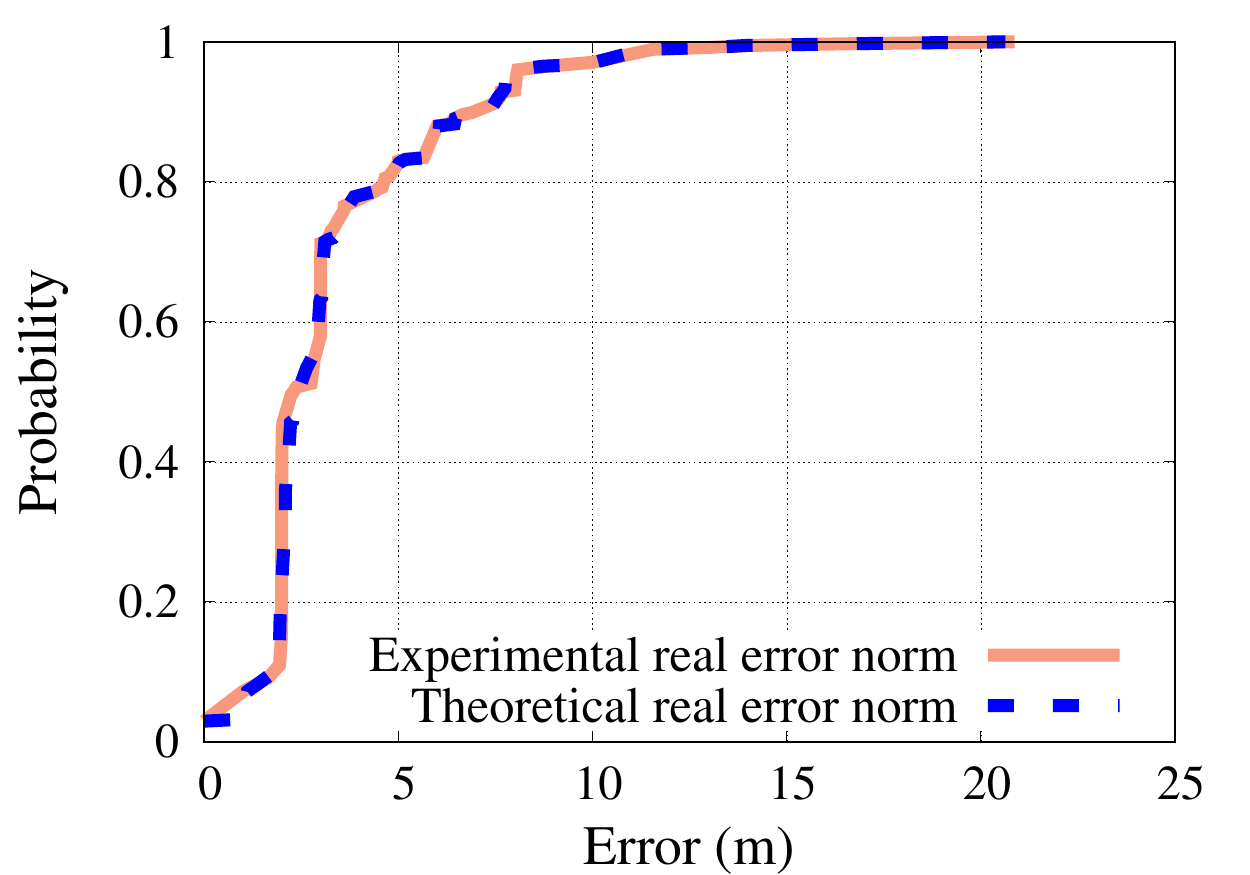}}
\caption{A comparison between experimental and theoretical real error norm CDFs (from Algorithm~\ref{algo:marking}).}
\label{fig:cdf_exp_thr}
\end{figure}

\subsubsection{Map Error}

To study the effect of applying Algorithm~\ref{algo:map} on the localization accuracy, we shift the real ground truth locations with different values from 1m to 6m in X and Y directions, simulating offset errors in the provided map.
Figure~\ref{fig:algo2} compares the validation error norm, theoretical real error norm, and experimental real error norm for the different map shift values. The validation error is obtained by comparing the system estimated location with the shifted ground truth locations in the collected dataset. Hence, it increases with the increase of the map shift. The experimental real error is obtained by comparing the system estimated location with the real ground truth locations in the collected dataset. Finally, the theoretical real error is obtained from Algorithm~\ref{algo:map} given the map shift value and the validation error statistics.

The figure shows that the calculated theoretical real error is closer to the experimental real error than the validation error norm which is typically used to evaluate the localization algorithms.
In particular, Algorithm~\ref{algo:map} does provide a more accurate estimate of the median and tail localization error by more than 20\% and 5\%, respectively, when the map is shifted by 1m in X and Y directions. This enhancement increases to 150\% and 72\%, respectively, when the map is shifted by 6m. Therefore, this algorithm can be used to obtain a more realistic quantification of the localization system error.

\begin{table}
\centering
\caption{Summary statistics for the different error metrics. The experimental results are calculated directly from the data while the theoretical real error is calculated from Algorithm~\ref{algo:marking}.} %
\label{table:exp_th_err}
\begin{tabular}{|c|c|c|c|} \hline
	\centering
    Metric & Marking error & Real error & Real error \\
    Type & Experimental & Experimental & Theoretical \\ \hline
	~ Mean ~ & 0.17 & 3.39 & ~ 3.40 (\textbf{0.3\%}) ~ \\
	~ 0.25 Q ~ & 0.11 & 2.00 & ~ 2.02 (\textbf{1.0\%}) ~ \\
	~ Median ~ & 0.15 & 2.37 & ~ 2.46 (\textbf{3.8\%}) ~ \\
	~ 0.75 Q ~ & 0.24 & 3.60 & ~ 3.70 (\textbf{2.8\%}) ~ \\
	~ 0.95 Q ~ & 0.30 & 7.99 & ~ 8.04 (\textbf{0.6\%}) ~ \\ \hline
\end{tabular}
\end{table}

\begin{figure*}[!t]
        \centering
        \begin{subfigure}[b]{0.3\linewidth}
                \includegraphics[width=0.99\linewidth]{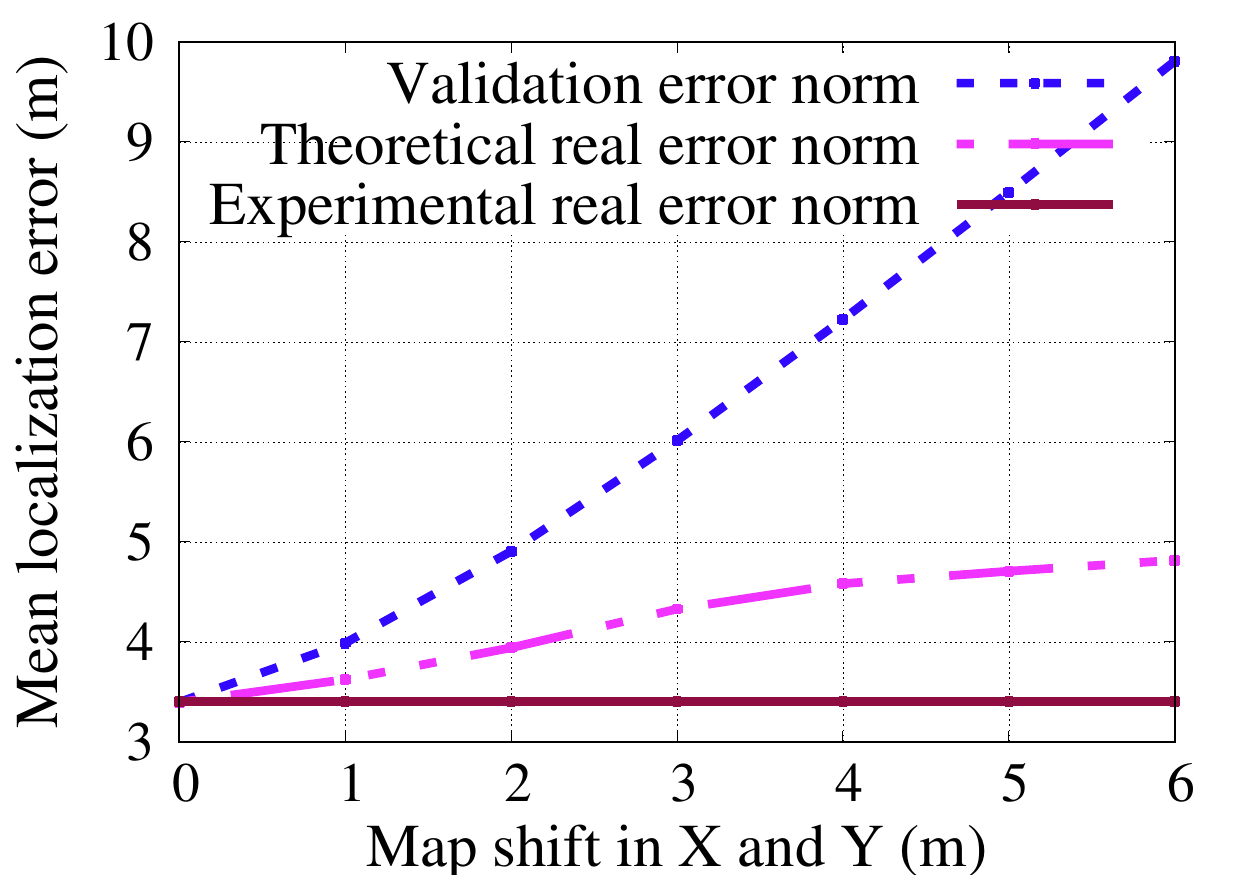}
           \caption{Mean.}
           \label{}
        \end{subfigure}%
        \begin{subfigure}[b]{0.3\linewidth}
                \includegraphics[width=0.99\linewidth]{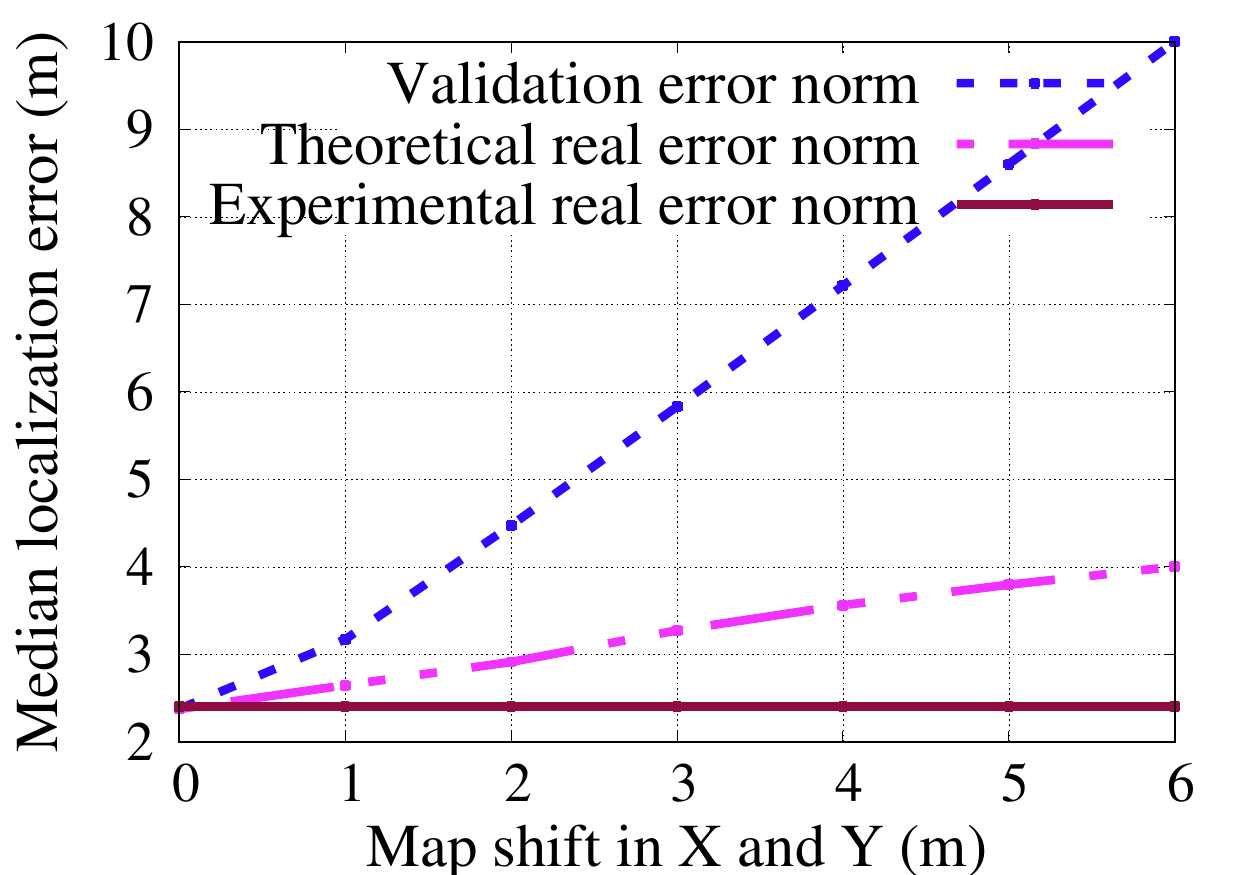}
           \caption{Median.}
           \label{}
        \end{subfigure}%
        \begin{subfigure}[b]{0.3\linewidth}
                \includegraphics[width=0.99\linewidth]{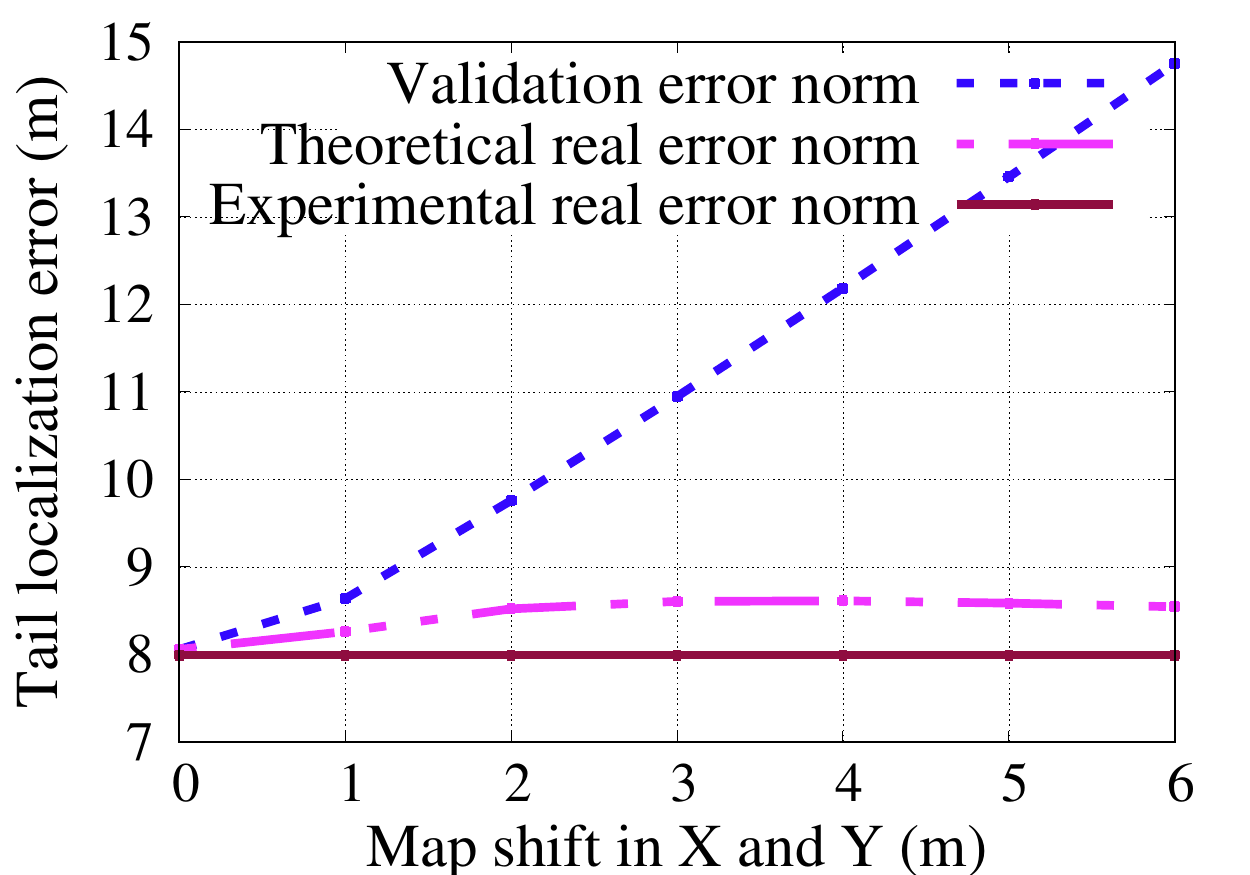}
              \caption{Tail.}
            \label{}
        \end{subfigure}
        \caption{A comparison between the validation error norm, theoretical real error norm (from Algorithm~\ref{algo:map}), and experimental real error norm for the different map shift values.} %
             \label{fig:algo2}
\end{figure*}

\section{Related Work}
\label{related_work}

In this section, we present a brief overview of the previous work on evaluating the accuracy of localization systems that is most relevant.

\subsection{AP Density Analysis}

In literature, there are studies that focus on the effect of access points (APs) number and position on the localization accuracy \cite{jia2019selecting, tian2020optimizing}. In \cite{jia2019selecting}, the authors proposed the theoretical error analysis for WiFi-based localization. It is shown that the Cramer-Rao lower bound (CRLB) \cite{smith2000intrinsic, van2004optimum} is inversely proportional to the number of the RSS measurements available from APs. They show that it is not necessary to leverage all the available APs to achieve the best accuracy. %
In \cite{tian2020optimizing}, authors proposed an algorithm that searches for the optimal APs and beacons placement that maximizes the localization accuracy. They tackle the problem of optimizing AP and beacon placement, which is NP-Complete, by proposing a heuristic differential evolution algorithm based on the widely used CRLB.

On the other hand, in this work, our goal is different: to get the real error statistics from the marking and map errors. Moreover, unlike these studies, which are limited to APs-based localization systems, our analysis is general for any localization system.

\subsection{Error Bounds Analysis}

Another set of studies put bounds on the localization error. Examples include time-of-arrival \cite{qi2006time}, time-difference-of-arrival \cite{dersan2002passive}, angle-of-arrival \cite{dersan2002passive}, and the RSS fingerprinting \cite{chang2004estimation, venkatesh2008multiple, patwari2002location}. The CRLB is widely used in these studies to put a bound on localization error \cite{chang2004estimation, qi2006time, hossain2010cramer}. In \cite{chang2004estimation}, authors studied the CRLB for anchored and anchor-free localization using noisy range measurements. They gave a method to compute the CRLB in terms of the geometry of the sensor network. In \cite{qi2006time}, the authors investigate the improvement in positioning accuracy if all multipath delays are processed (instead of using the first path). The authors show that using the first arrival only is sufficient for optimal localization when there is no prior information about the non-line-of-sight delays. When such prior information is available, the multipath delays can improve the localization accuracy. The best achievable accuracy is evaluated in terms of CRLB and the generalized-CRLB. In \cite{hossain2010cramer}, the authors analyze the CRLB of localization using signal strength difference as location fingerprint.

In contrast, our target is not to put a bound on the localization error only, but to establish bounds on different metrics between the localization error and ground truth errors.

\subsection{Localization Confidence Estimation}

Recently, different confidence estimation techniques are proposed for different localization systems. Examples include GPS \cite{drawil2012gps,moghtadaiee2012accuracy}, GNSS \cite{niu2014using} and indoor localization techniques such as \cite{lemelson2009error, moghtadaiee2012accuracy, dearman2007exploration}. Confidence estimation for the GPS is typically derived from the geometric dilution of precision which measures the confidence as a function of error caused by the geometry of the GPS satellites
\cite{drawil2012gps, moghtadaiee2012accuracy, niu2014using, dearman2007exploration}. Authors of \cite{drawil2012gps, niu2014using} further analyze the error characteristics of GPS and GNSS localization systems and derive an error model to estimate the localization error based on a combination of the number of satellites, dilution of precision, received signal strength and receiver speed.
In \cite{dearman2007exploration}, authors maintain a database of locations and their corresponding error measurement to estimate the localization accuracy.
In \cite{niu2014using}, authors proposed a technique that identifies the dominant noise types and builds an error source model to estimate the GNSS positioning error based on the Allen variance method \cite{el2007analysis}.
In \cite{elbakly2016cone}, authors estimate the indoor localization confidence by assuming that the error in user location follows a Gaussian distribution. They proposed a system that can work in real-time to get the confidence in the estimated location from the history of the previous estimated locations.

In comparison to these systems, which do not consider the marking and map errors, we provide a general framework to handle both kinds of ground truth errors.

\section{Conclusions}
\label{conc}

In this paper, we presented a theoretical framework for analyzing the effect of ground truth errors on the evaluation of localization systems. We designed two algorithms for computing the real algorithmic error from the validation error and marking/map ground truth errors, respectively. We showed that the impact of marking error is quadratic in its ground truth error, and inversely proportional to validation error. We further established bounds on different performance metrics: that the 95\%-tail error is at least twice of the median/mean errors and proved that marking error has more impact than map error on the mean and median, but less impact on the tail.

We validated our theoretical assumptions and analysis on a real indoor WiFi dataset. Our experiments show the ability of our analysis to obtain a more realistic localization error in the presence of ground truth errors. Specifically, we showed that Algorithm~\ref{algo:marking} matches the real error within 4\% in all percentiles, and Algorithm~\ref{algo:map} provides a more accurate estimate of the median and tail errors by more than 150\% and 72\%, respectively, when the map is shifted by 6m.

For future work, we are extending our analysis to address map scale errors and modifying the model to handle asymmetric X/Y error distributions.

\section*{Acknowledgment}

The authors wish to thank Patrick Robertson for proposing the initial questions to be considered in this paper.

\end{document}